\documentclass[12pt]{amsart}

\usepackage{graphicx}
\usepackage[all]{xy}

\usepackage{amsaddr}
\usepackage{amsmath}
\usepackage{amssymb}
\usepackage{amsfonts}
\usepackage{bm}
\usepackage{color}
\usepackage{complexity}

\usepackage[pdfpagelabels,pdftex,bookmarks,breaklinks]{hyperref}
\hypersetup{colorlinks, linkcolor=black, citecolor=black}

\usepackage{amsthm}

\newcommand{\be}{\begin{equation}}
\newcommand{\ee}{\end{equation}}
\newcommand{\ba}{\begin{array}}
\newcommand{\ea}{\end{array}}
\newcommand{\bea}{\begin{eqnarray}}
\newcommand{\eea}{\end{eqnarray}}

\newcommand{\cal}{\mathcal}

\newcommand{\calL}{{\cal L }}

\newcommand{\calZ}{{\cal Z }}
\newcommand{\calW}{{\cal W }}

\newcommand{\calS}{{\cal S }}
\newcommand{\calI}{{\cal I }}

\newcommand{\CC}{\mathbb{C}}

\newcommand{\RR}{\mathbb{R}}

\newcommand{\ra}{\rangle}

\newcommand{\prob}[1]{\mathrm{Pr}{\left[\, {#1} \, \right]}}
\newcommand{\pois}[1]{\mathrm{Pois}{\left[ {#1}  \right]}}
\newcommand{\expect}[1]{{\mathbb{E}}{\left[ {#1}  \right]}}

\newcommand{\trace}[1]{{\mathrm{Tr}{#1}}}

\newtheorem{dfn}{Definition}

\newtheorem{lemma}{Lemma}

\newtheorem{theorem}{Theorem}

\begin{document}

%\title{When simulation of stoquastic Hamiltonians is easy}
%\title{Towards efficient simulation of stoquastic Hamiltonians}
%\title{Two examples of quantum Hamiltonians that can be simulated classically}
%\title{When simulation of stoquastic Hamiltonians is easy}
%\title{Simulating stoquastic Hamiltonians with randomized algorithms}
\title{Monte Carlo simulation of stoquastic Hamiltonians}

\author{Sergey Bravyi}
\address{IBM T.J. Watson Research Center\\
Yorktown Heights, NY 10598}
\email{sbravyi@us.ibm.com}

\date{}

\begin{abstract}
Stoquastic  Hamiltonians  are
characterized by the property that their off-diagonal matrix elements
in the standard product basis are real and non-positive.  Many interesting
quantum models fall into this class including the  Transverse field Ising Model (TIM),
 the Heisenberg model on bipartite graphs, and the bosonic Hubbard
model. Here we  consider the problem  of estimating the ground state energy of a
local stoquastic  Hamiltonian $H$ with a promise that the ground state of $H$ has a 
non-negligible correlation with some `guiding' state 
that admits a concise classical description. A formalized
version of this problem called Guided Stoquastic Hamiltonian  is shown to be
complete for the complexity class $\MA$ (a probabilistic analogue of $\NP$).
To prove this result  we employ the Projection Monte Carlo algorithm with
a variable number of walkers.
Secondly, we  show that the ground state and thermal equilibrium properties
of the ferromagnetic TIM can be simulated in polynomial time on a classical
probabilistic computer. 
 This result is based on the approximation algorithm
for the classical ferromagnetic Ising model due to Jerrrum and Sinclair (1993).
\end{abstract}

\maketitle

\section{Introduction}

Calculating the ground state energy and thermal equilibrium properties
of interacting quantum many-body systems is one of the central problems
in quantum chemistry and condensed matter physics. It was realized early on 
that the computational complexity of this problem is strongly affected by statistics of the
constituent particles. For systems composed of bosons and for certain special classes of spin Hamiltonians
the quantum partition function can be mapped to the one of a classical 
system occupying one extra spatial dimension~\cite{Sachdev2007}
which often enables efficient Monte Carlo simulation~\cite{suzuki1977,prokof1996exact,sandvik1991,Trivedi1989,Buonaura1998,Wessel2004}.
On the other hand, for systems composed of fermions
 and for the vast majority of spin Hamiltonians
the quantum-to-classical mapping produces a partition function with
unphysical Boltzmann weights taking both positive
and negative (or complex) values --- a phenomenon known as the ``sign problem".
In certain special cases, such as the Heisenberg antiferromagnetic model on bipartite graphs, 
the sign problem can be avoided by a suitable basis change~\cite{Trivedi1989}. In general however
the sign problem appears to be an intrinsic feature of the quantum mechanics
that limits applicability of Monte Carlo simulation algorithms~\cite{Troyer2005}.

The present paper attempts to provide a rigorous basis for the common belief that
Hamiltonians  avoiding the sign problem are ``easy" to simulate. 
We focus on systems composed of qubits (spins-$1/2$) and 
study the class of so-called {\em stoquastic}\footnote{The term stoquastic  introduced in Ref.~\cite{BDOT06}
conveys  the fact that the studied problems lie on the border between quantum mechanics and classical theory
of stochastic matrices.  For non-technical purposes `stoquastic' is equivalent to
``avoiding the sign problem". } Hamiltonians~\cite{BDOT06}.
The defining property of stoquastic Hamiltonians is that their off-diagonal matrix elements in the computational 
basis must be real and non-positive. 
More formally, let $n$ be the number of qubits and $k=O(1)$ be a small constant. 
Define a $k$-local stoquastic Hamiltonian as 
\begin{equation}
\label{H1}
H=\sum_{\alpha=1}^M  H_\alpha,
\end{equation}
where each term $H_\alpha$ is a hermitian operator on $n$ qubits that acts non-trivially on 
a subset of at most $k$ qubits and satisfies 
\begin{equation}
\label{Hxy}
\langle x|H_\alpha|y\rangle \le 0 \quad \mbox{for all $x,y\in \{0,1\}^n$ with $x\ne y$}.
\end{equation}
From Eq.~(\ref{Hxy}) one can easily infer that the Boltzmann exponential
operator $e^{-\beta H_\alpha}$ has non-negative matrix elements for any
inverse temperature $\beta\ge 0$. Furthermore, the matrix element
$\langle x|e^{-\beta H_\alpha}|y\rangle$ depends non-trivially only on $O(1)$ bits of $x$ and $y$.
Combining these two properties one can
approximate the quantum partition function $\trace{\, e^{-\beta H}}$
by a classical partition function that involves local non-negative Boltzmann weights,
for example, using the Suzuki-Trotter formula.  
In that sense any stoquastic Hamiltonian avoids the sign problem.

Below we assume that the Hamiltonian is normalized\footnote{Here 
$poly(n)$ and $poly(1/n)$ denote functions of $n$ with the asymptotic scaling $n^{O(1)}$ 
and $n^{-O(1)}$.}
 such that 
$\|H_\alpha\|\le poly(n)$ for all $\alpha$ and $M\le poly(n)$. 
The quantity we are interested in is the ground state energy 
\begin{equation}
\label{lambda}
\lambda=\min_\psi \langle \psi|H|\psi\rangle, \quad {\langle \psi|\psi\rangle=1}.
\end{equation}
We shall consider a problem Stoquastic Local Hamiltonian (Stoq-LH) where
the goal is to estimate the ground state energy with a small additive error.
More formally, an instance of Stoq-LH includes the number of qubits $n$,
a list of interactions $H_\alpha$ as above, and real numbers $\lambda_{yes}<\lambda_{no}$
such that $\lambda_{no}-\lambda_{yes}\ge poly(1/n)$.
The problem is to decide whether $\lambda\le \lambda_{yes}$ (yes-instance) or $\lambda\ge \lambda_{no}$
(no-instance)  given a
promise that $\lambda$ does not belong to the interval $(\lambda_{yes},\lambda_{no})$.

Since any Hamiltonian diagonal in the standard basis is stoquastic, Stoq-LH encompasses
hard classical optimization problems such as $k$-SAT or MAX-CUT. This shows that Stoq-LH
is at least $\NP$-hard and essentially rules out a possibility that Stoq-LH admits an efficient
algorithm. A natural next question is whether  Stoq-LH is contained in the class $\NP$
or its probabilistic analogue called $\MA$ (Merlin-Arthur games)~\cite{Babai1985}. Loosely speaking, proving
the containment in $\NP$ or $\MA$ would imply that certain ground state properties of
stoquastic Hamiltonians can be efficiently {\em verified} 
even though they cannot be efficiently computed. Some progress along these lines
has been made in Refs.~\cite{BDOT06,BT07} by  proving that Stoq-LH
is contained in $\MA$ for the special case of frustration-free Hamiltonians
(for a detailed discussion of the previous work see Section~\ref{subs:previous}).
The present paper extends these results by identifying two new classes
of stoquastic Hamiltonians whose ground state properties can be efficiently
verified and, in certain cases, efficiently computed.

\section{Summary of results}

The first class of Hamiltonians that we study is motivated by the 
Quantum Phase Estimation (QPE) algorithm~\cite{KSV}
and the question of whether application of QPE  to stoquastic Hamiltonians
can be efficiently simulated classically. 
Recall that  the ground state energy of a local Hamiltonian $H$ can be efficiently estimated
on a quantum computer  via QPE only if one is able to prepare some initial state $\phi$ 
that has a non-negligible (at least $poly(1/n)$) overlap  with the exact ground state of $H$. 
Following Ref.~\cite{CerfOlivier95}
we shall refer to such initial state $\phi$ as a {\em guiding state} since its purpose is 
to guide the algorithm towards the ground state of $H$.  
It is usually assumed that a good choice of  the guiding state can be made if the system
under consideration is sufficiently well understood and some physical theory describing,
at least approximately, its ground state properties is available. This
motivates the study of  the Stoquastic Local Hamiltonian problem  with an extra promise that 
the Hamiltonian admits a guiding state. A natural question is whether in this case QPE can be replaced by some classical algorithm.
We show that the answer is YES provided that the guiding state has 
efficiently computable amplitudes and
a non-negligible pointwise correlation with the exact ground state as formally defined below. 
%In the case of frustration-free stoquastic Hamiltonians a stronger result was obtained
%in Ref.~\cite{BT07} by showing that a stoquastic version of the quantum $k$-SAT problem
%is contained in  promise-$\MA$ for any $k=O(1)$ and is complete for promise-$\MA$
%for $k=6$. The property of being frustration-free however
%does not appear naturally in physically relevant quantum models. 

\begin{dfn}
\label{dfn:guided}
Let $H$ be a  stoquastic Hamiltonian. We will say that $H$ admits a guiding state iff there exists a pair of 
normalized $n$-qubit states $\psi,\phi$ with
 non-negative amplitudes in the standard
basis such that  $\psi$ is a ground state of $H$,
the function $x\to \langle x|\phi\rangle$ is computable by a classical circuit
of size $poly(n)$, and 
\begin{equation}
\label{guiding}
\langle x|\phi\rangle \ge \frac{\langle x|\psi\rangle}{poly(n)} \quad \mbox{for all $x\in \{0,1\}^n$}
\end{equation}
\end{dfn}
A state $\phi$ satisfying the above conditions will be referred to as a guiding state. 
Note that any stoquastic Hamiltonian has a ground state with non-negative amplitudes due to the 
Perron-Frobenius theorem. However, it is rather unlikely that any stoquastic Hamiltonian admits
a guiding state. 
It should be emphasized that our classical algorithm  and QPE 
need guiding states with different properties. The pointwise correlation condition in Eq.~(\ref{guiding})
is much stronger that the non-negligible overlap condition needed for QPE. Note that
Eq.~(\ref{guiding}) implies $\langle \psi|\phi\rangle\ge poly(1/n)$, but the converse is not true. 
On the other hand, given a  short classical circuit that computes amplitudes of $\phi$, 
generally one cannot convert it to  a short quantum circuit that prepares $\phi$.

Define a problem Guided Stoquastic Local Hamiltonian (Guided Stoq-LH) as a special case of Stoq-LH where any yes-instance
must satisfy two promises:
 {\em Promise~1:} $\lambda\le \lambda_{yes}$ and
{\em Promise~2:}  $H$ admits a guiding state.
In the case of no-instances the only promise is that $\lambda\ge \lambda_{no}$. 
Note that the guiding state is not regarded as a part of the input.
Our main result is the following.
 \begin{theorem}
 \label{thm:guided}
Guided Stoq-LH  is contained in the class promise-$\MA$ for any constant $k$.
Guided Stoq-LH is complete for promise-$\MA$ for $k\ge 6$. 
\end{theorem}
Less formally, Theorem~\ref{thm:guided} asserts that a verifier (usually called Arthur)
with polynomial classical computational resources and a random number generator
can reliably distinguish between
yes- and no-instances of the problem by consulting an untrusted 
prover (usually called Merlin) which has  unlimited computational power.
Merlin's goal is to convince Arthur that a given instance of the problem is positive, 
that is, $\lambda\le \lambda_{yes}$.  To this end Merlin sends Arthur a witness
--- a classical bit string which, if Merlin is honest,  includes the description 
of a guiding state and certain additional information.
In the case of yes-instances Merlin can always find a witness
convincing Arthur that $\lambda\le \lambda_{yes}$ with probability close to one. 
Meanwhile,  for  no-instances Arthur decides that $\lambda\ge \lambda_{no}$ 
with probability close to one for any Merlin's witness.

To prove the containment in $\MA$ we employ a version of the Projection Monte Carlo 
algorithm with a variable number of walkers. This algorithm has been previously proposed in the context of  quantum 
Monte Carlo simulations by Cerf and Martin~\cite{CerfOlivier95} and by
Oliveira~\cite{Roberto}. The key idea of the algorithm
is to convert  a stoquastic Hamiltonian into a random walk using Poisson-distributed random variables.
A state of the walk is a function that assigns a non-negative integer to each $n$-bit binary string. 
Such  function can be visualized as a population of walkers distributed over
the Boolean cube. A typical step of the walk involves moving one or several walkers to a new
location, creating  new walkers, and eliminating some existing walkers. 
We show that for any yes-instance Merlin can choose a witness
such the the total population size is confined to the interval $[1,poly(n)]$ during all steps of the walk
with a non-negligible probability. Meanwhile, for any no-instance and for any Merlin's
witness the average population size  decreases exponentially with time. In this case
 the population either becomes extinct after $poly(n)$ steps or  becomes
 too large at some intermediate step 
 due to statistical fluctuations. By  implementing the random walk and monitoring the population size 
 Arthur can therefore distinguish between yes- and no-instances.
Our rigorous analysis of the algorithm  based on the second moment method appears to be new. 
The last statement of Theorem~\ref{thm:guided} 
($\MA$-completeness) follows trivially  from the results of Ref.~\cite{BDOT06}.
We discuss some open problem and potential improvements of Theorem~\ref{thm:guided}  in Section~\ref{subs:problems1}.

%%%%%%%%%%%%%%%%%%%%%%%%TIM%%%%%%%%%%%%%%%%%%%%%%%%%%%%%%
The second  class of  Hamiltonians that we study is 
 the Transverse field Ising Model (TIM):
\begin{equation}
\label{TIM1}
H=-\sum_{1\le u<v\le n} J_{u,v} Z_u Z_v -\sum_{1\le u\le n} h_u X_u.
\end{equation}
Here $X_u$ and $Z_u$ are the Pauli operators acting on a qubit $u$, while
$J_{u,v}$ and $h_u$ are real coefficients. 
A direct inspection shows that $H$ is stoquastic iff $h_u\ge 0$ for all $u$. 
Any TIM Hamiltonian can be made stoquastic by a transformation
$H\to Z_u H Z_u$ that flips the sign of $h_u$ without changing any other terms. 
Define a partition function 
\begin{equation}
\label{partition}
{\calZ}=\trace{\, e^{-H}}.
\end{equation}
Our second result shows that $\calZ$ can be efficiently approximated with a small multiplicative error
$\delta$  in the special case of the {\em ferromagnetic} TIM, that is, when $J_{u,v}\ge 0$ for all $u,v$. 
More precisely, let
\begin{equation}
\label{norm}
J=\max{\{ J_{u,v}, |h_u|\}}
\end{equation}
be the maximum norm of the interactions and $0<\delta<1$ be the desired precision. 
We shall say that $\calZ$ admits a {\em fully polynomial randomized approximation scheme} (FPRAS)
if there exists a classical randomized algorithm with the running time $poly(n,J,\delta^{-1})$ that takes as input
a pair $(H,\delta)$ and outputs a random variable $\tilde{\calZ}$ such that 
\begin{equation}
\label{FPRAS}
\prob{ (1-\delta)\calZ\le \tilde{\calZ}\le (1+\delta)\calZ } \ge 2/3.
\end{equation}
\begin{theorem}
\label{thm:TIM}
The partition function of the ferromagnetic TIM admits FPRAS. 
\end{theorem}
This immediately implies that the free energy $F(T)=-T\log{\left(\trace{\, e^{-H/T}}\right)}$ can be approximated
with an additive error $\delta$ in time $poly(n,J,\delta^{-1},T^{-1})$.
Furthermore, since $F(0)-F(T)=\int_0^T dT' S(T')\le nT$, where $S(T')$ is the entropy
of the Gibbs state, we conclude that the ground state energy $\lambda=F(0)$ can be 
approximated with an additive error $\delta$ in time $poly(n,J,\delta^{-1})$.
Theorem~\ref{thm:TIM} is a simple application of the seminal result by Jerrum and Sinclair~\cite{JS93}
who showed that thermal equilibrium properties of the ferromagnetic  {\em classical}  Ising model ($h_u=0$ for all $u$)
can be simulated efficiently.
\begin{theorem}[\bf Jerrum and Sinclair~\cite{JS93}]
\label{thm:IM}
The partition function of the ferromagnetic classical Ising model admits FPRAS.
\end{theorem}
More precisely, the algorithm proposed in Ref.~\cite{JS93} has running time $poly(n,\delta^{-1})$
which is independent on $J$,
as long as the cost of arithmetic operations with $J_{u,v}$ can be neglected\footnote{Intuitively, the lack of dependence on $J$ reflects the fact that
zero-temperature properties of the ferromagnetic Ising model are trivial --- all spins 
are oriented in the same direction. Since we have included  the inverse temperature into
the coefficients $J_{u,v}$, the limit $J\to \infty$ corresponds to the zero temperature.
Similarly, if $J_{u,v}\to \infty$ for some pair of spins $u,v$, the problem with $n$ spins
can be reduced to the problem with $n-1$ spins by merging $u$ and $v$ into one effective spin. 
}.
To go from Theorem~\ref{thm:IM} to Theorem~\ref{thm:TIM} we employ the standard quantum-to-classical mapping 
based on the Suzuki-Trotter formula. The only new ingredient that we add is a proof
that the Suzuki-Trotter approximation
leads to a small multiplicative  error (as opposed to the additive error usually studied in the literature). 
We emphasize that although the proposed FPRAS for TIM is efficient in the complexity theory sense, it
can hardly be used in practice. A rigorous upper bound on the running time
of the FPRAS obtained in  the proof of Theorem~\ref{thm:TIM}
is $O(n^{59} J^{21} \delta^{-9})$. Clearly, this leaves a lot of room for improvements.
%We anticipate that the running time can be substantially improved using
%by combining the algorithm of Ref.() with the acceptance ratio method due to Bennett
%for estimating the ratio of partition functions at two close temperatures

\section{Discussion and previous work}
\label{subs:previous}

The class of stoquastic Hamiltonians encompasses many interesting
quantum spin models originated both from the condensed matter physics and the
quantum computing field. Well-known examples
include TIM, the Heisenberg ferromagnetic and antiferromagnetic models
(the latter can be made stoquastic on any bipartite graph), quantum annealing
Hamiltonians~\cite{Farhi2001quantum,Somma2012}, the toric code Hamiltonian~\cite{Kitaev2003fault},
and Hamiltonians derived from reversible Markov 
chains~\cite{Aharonov2003adiabatic,Henley2004classical,Verstraete2006criticality}.
The definition of stoquastic Hamiltonians can be naturally extended to higher-dimensional spins and 
to bosonic systems. Notable examples of models in this category are the quantum double Hamiltonian~\cite{Kitaev2003fault},
bosonic Hubbard model~\cite{Wessel2004}, 
and Hamiltonians describing ``flux-type" Josephson junction qubits~\cite{Devoret2004}.
Identifying ``easy" and ``hard" instances of stoquastic Hamiltonians is therefore
important as it could give insights on the power and limitations of quantum
Monte Carlo algorithms~\cite{Hetherington1984,Hastings2013} and
contribute to our understanding of 
speedups in quantum annealing algorithms~\cite{Farhi2009quantum,Farhi2012}.

Complexity of stoquastic Hamiltonians has been partially characterized
in Ref.~\cite{BDOT06} by proving that Stoq-LH is hard for the complexity class $\MA$ 
and contained in the class $\AM$.
Here $\MA$ and $\AM$ are probabilistic analogues of $\NP$ with 
one and two rounds of communication between the prover and the verifier respectively~\cite{Babai1985}.
It was shown~\cite{BDOT06} that 
the complexity of Stoq-LH does not depend on the locality parameter $k$
as long as $2\le k\le O(1)$. A closely related problem of verifying 
consistency of local reduced density matrices with non-negative matrix elements
was studied by Liu~\cite{Liu2007}.
Finally, Stoq-LH is contained in the class $\QMA$ (a quantum
analogue of $\NP$) since estimating the ground state energy of a general local Hamiltonian
is known to be $\QMA$-complete problem~\cite{KSV}. 
Interestingly, the problem of estimating the {\em largest} eigenvalue of a local stoquastic
Hamiltonian was shown to be $\QMA$-complete by Jordan et at~\cite{Jordan2010}.

TIM occupies a special place in the family of stoquastic Hamiltonians due to its simplicity
and a vast body of work devoted to it. 
In particular, TIM defines one of the four 
classes in the complexity classification of $2$-local quantum Hamiltonians
developed recently by Cubitt and Montanaro~\cite{Cubitt2013complexity}
(with the other three classes being $\P$, $\NP$, and $\QMA$). 
Experimental implementation of quantum annealing algorithms 
based on TIM  has been demonstrated~\cite{Boixo2013}, see also~\cite{Ronnow2014,Shin2014quantum}
for possible interpretation of these experiments. 
It was shown recently that any $k$-local stoquastic Hamiltonian can be represented as an effective low-energy
theory emerging from TIM on a constant-degree graph~\cite{BH14}. As a consequence, 
Stoq-LH for TIM Hamiltonians is as hard as Stoq-LH for general  stoquastic Hamiltonians.
Theorem~\ref{thm:TIM} implies that a presence of antiferromagnetic spin couplings
in a TIM Hamiltonian is essential for this hardness result. 

To put Theorem~\ref{thm:TIM} in a broader context, let us briefly discuss the previous work
on complexity of the {\em classical} ferromagnetic  Ising model with local magnetic fields,
\[
H=-\sum_{1\le u<v\le n} J_{u,v} Z_u Z_v -\sum_{1\le u\le n} h_u Z_u, \quad \quad J_{u,v}\ge 0.
\]
In the case when the local  fields are uniform, i.e. $h_u\ge 0$ or $h_u\le 0$ for all $u$,
the minimum energy problem is trivial since the ground state is given by
$Z_u=+1$ or $Z_u=-1$ for all $u$ respectively. 
In the general case when the local fields may take both positive and negative values,
the minimum energy of $H$ can be computed in time $O(n^3)$ by a reduction to the Maximum Flow
problem~\cite{Rieger97}. Suppose now that ones goal is to compute the partition function
${\calZ}=\trace{\, e^{-H}}$. 
In the case of uniform  local fields $\calZ$ is known to  admit FPRAS while
an exact computation of $\calZ$ is  $\#P$-hard~\cite{JS93}.
In the general case when the local fields may have both positive and negative signs, 
approximating $\calZ$ with a small multiplicative error is known to be as
hard as the approximating the number of independent sets in a bipartite graph which is unlikely to have
a polynomial time algorithm~\cite{GoldbergJerrum2005}. 
Finally, we note that if one does not insist on rigorous bounds on the running time, 
there exist alternative more practical algorithms for approximating the partition function $\calZ$,
such as the Swendsen-Wang algorithm~\cite{SW87}.
The latter is known to have the running time growing exponentially with $n$ in the case
of the ferromagnetic $3$-state Potts model~\cite{Gore1999}.

\section{Guided Stoquastic Hamiltonians}

In this section we prove Theorem~\ref{thm:guided}.
We start from the containment in $\MA$ which is by far  the most difficult part. 
 We shall describe a classical probabilistic algorithm that takes as input
a problem instance  $\calI=(n,H=\sum_{\alpha} H_\alpha,\lambda_{yes},\lambda_{no})$
and a witness string $\calW$ of length $poly(n)$. The algorithm runs in time $poly(n)$ and outputs
`accept' or `reject'. Let $P_{acc}=P_{acc}(\calI,\calW)$ be the acceptance probability.
The statement that  Guided Stoq-LH is contained
in promise-$\MA$ is equivalent to the following conditions.\\
{\bf Completeness:} If $\calI$ is a yes-instance then $P_{acc}(\calI,\calW)\ge 2/3$ for some witness $\calW$.\\
{\bf Soundness:} If $\calI$ is a no-instance then $P_{acc}(\calI,\calW)\le 1/3$ for any witness $\calW$.\\
Following the standard terminology, we shall refer to the party running the algorithm as Arthur
and the party providing the witness as Merlin. 

Let $(n,H=\sum_{\alpha} H_\alpha,\lambda_{yes},\lambda_{no})$ be an instance of Guided Stoq-LH.
Let $\lambda$ be the ground state energy of $H$. One can easily check that 
$-J\le \lambda \le J$, where 
 $J\equiv \sum_\alpha \|H_\alpha\|$. 
 We can assume that  $-J\le \lambda_{yes}<\lambda_{no}\le J$  since otherwise the problem becomes trivial. 
To convince Arthur that  $\lambda\le \lambda_{yes}$ 
Merlin will present a witness that consists of three parts:
\begin{itemize}
\item A real number $-J\le \lambda_M\le \lambda_{yes}$
\item A classical circuit that computes
some function $\phi_M\, : \, \{0,1\}^n \to \RR$.
\item A binary string $x_M\in \{0,1\}^n$
\end{itemize}
Let us agree that for a yes-instance $\lambda_M=\lambda$ is the ground state energy of $H$
while the function $\phi_M$ computes amplitudes of some guiding  state for $H$.
The string $x_M$ must satisfy several technical conditions  stated below.
For a no-instance Merlin may try to cheat, that is,  $\lambda_M, \phi_M, x_M$  could be arbitrary. 
We shall assume that Arthur rejects the witness right away if it does not fit the specified format,
for example, if $\lambda_M>\lambda_{yes}$.

\subsection{Verification algorithm}
To define Arthur's verification algorithm  we first  convert $H$ into a sparse non-negative matrix $G$ such that 
$\|G\|=1$ for yes-instances and $\|G\| <1$ for no-instances. 
To this end choose  $\beta=1/(2J)$ and define  
\begin{equation}
\label{GreenFunction}
G=I-\beta(H-\lambda_M I).
\end{equation}
Note that $\| \beta(H-\lambda_M I)\| \le \beta(\|H\| + |\lambda_M|) \le 2\beta J=1$. This shows that 
$G$ is a  positive semidefinite operator with real 
non-negative matrix elements in the standard basis  for both yes- and no-instances.
The operator norm of $G$ is equal to its largest eigenvalue, that is, 
\begin{equation}
\label{mu}
\|G\|=1-\beta(\lambda-\lambda_M).
\end{equation}
Furthermore,  $\psi$ is a ground state of $H$ iff $\psi$ is an eigenvector
of $G$ with the eigenvalue $\|G\|$. For a yes-instance $\lambda_M=\lambda$, that is
$\|G\|=1$. For a no-instance $\lambda\ge \lambda_{no}$ whereas $\lambda_M\le \lambda_{yes}$.
Therefore 
\begin{eqnarray}
\mbox{yes-instance} &\Rightarrow& \|G\|=1 \\
\mbox{no-instance} &\Rightarrow& \|G\| \le 1-\Delta, \label{Gnorm-no} 
\end{eqnarray}
where $\Delta$ is a ``decision gap" defined as
\begin{equation}
\label{Delta}
\Delta=\beta(\lambda_{no}-\lambda_{yes})\ge poly(1/n).
\end{equation}

The next step is to convert $G$ into  a random walk. Here we adopt 
the Projection Monte Carlo method with a variable number of walkers 
 which has been previously proposed in the context of  Quantum Monte Carlo simulations~\cite{CerfOlivier95,Roberto}.
(As was pointed out in Ref.~\cite{CerfOlivier95}, this method has a more favorable scaling
of statistical fluctuations in comparison to more widely used Green's function
 Monte Carlo~\cite{Hetherington1984,Trivedi1989,Buonaura1998}.)
A state of the walk is defined as a function 
\begin{equation}
\label{gamma}
\gamma\, : \, \{0,1\}^n \to \{0,1,2,\ldots \}
\end{equation}
that assigns a non-negative integer $\gamma(x)$ to each binary string $x\in \{0,1\}^n$.
Loosely speaking, the function $\gamma$ describes a population
of walkers distributed over points of the Boolean cube 
$\{0,1\}^n$. The meaning of $\gamma(x)$ is the occupation number of a point $x$.
We shall say that a point $x$ is {\em empty} or {\em occupied} if $\gamma(x)=0$ or $\gamma(x)\ge 1$ 
respectively. 
Our definition of the walk will depend on three parameters: the number of steps $L$, 
a cutoff population size $\Gamma_{max}$, and a cutoff 
guiding state amplitude $\phi_{min}$. We shall choose $L=poly(n)$, 
$\Gamma_{max}=poly(n)$, and $\phi_{min}=2^{-n-1}$.
A walk with $L$ steps is  a random sequence
of states $\gamma_0,\gamma_1,\gamma_2,\ldots,\gamma_L$.
We choose the initial state $\gamma_0$ as
\begin{equation}
\label{init}
\gamma_0(x)=\left\{ \begin{array}{rcl}
1 &\mbox{if} & x=x_M,\\
0 && \mbox{otherwise.} \\
\end{array}\right.
\end{equation}
Here  $x_M$ is the string received from Merlin.
Transition probabilities of the walk will be related to matrix elements of $G$. Define
a non-negative matrix $P$ of size $2^n$ such that 
\begin{equation}
\label{P(x,y)}
\langle x|P|y\rangle=\frac{\phi(y)}{\phi(x)} \, \langle x|G|y\rangle,
\end{equation}
where $\phi$ is a ``regularized version" of $\phi_M$ defined as
\begin{equation}
\label{regularized}
\phi(x)=\left\{ \begin{array}{rcl}
\phi_M(x) &\mbox{if}& \phi_{min}\le \phi_M(x)\le 1, \\
1 &\mbox{if} & \phi_M(x)>1, \\
\phi_{min} &\mbox{if} & \phi_M(x)<\phi_{min}. \\
\end{array}
\right.
\end{equation}
Recall that $\phi_{min}=2^{-n-1}$.
Note that the circuit computing $\phi_M(x)$ can be easily converted to the one
computing $\phi(x)$ without substantial increase in size. 
The matrix element  $\langle x|P|y\rangle$ will determine the rate
 at which walkers are created at a point $y$ at step $t+1$ per each walker located
at a point $x$ at step $t$. More formally, given a real number $p\ge 0$ let
$\pois{p}$ be a non-negative integer random variable drawn from the Poisson distribution with the mean $p$.
In other words, $k= \pois{p}$ iff $\prob{k}=e^{-p} p^k/k!$ for $k\ge 0$. 
Let us agree that $\pois{0}=0$ with probability one. 
For each $y\in \{0,1\}^n$ define
\begin{equation}
\label{next}
\gamma_{t+1}(y)=\sum_{x\in \{0,1\}^n} \pois{\gamma_t(x) \langle x|P|y\rangle},
\end{equation}
where all terms 
represent independent Poisson variables. We shall need these
well-known properties of the Poisson distribution:
\begin{equation}
\label{facts}
\expect{\pois{p}}=p, \quad \expect{\pois{p}^2}=p^2+p, \quad \pois{p}+\pois{q}=\pois{p+q}.
\end{equation}
The last equality involves a sum of two independent Poisson-distributed variables.
The total  population size at a step $t$ is
\begin{equation}
\label{N}
\Gamma_t=\sum_{x\in \{0,1\}^n} \gamma_t(x).
\end{equation}
%Note that $\Gamma_t$ takes only non-negative integer values. 
 Arthur's verification algorithm is defined as follows.
\begin{enumerate}
\item Receive $\lambda_M,\phi_M,x_M$ from Merlin. 
\item Make $L$ steps of the random walk defined above.
Abort and reject unless $\Gamma_t\le \Gamma_{max}$ for all steps $t$. 
\item Accept if $\Gamma_L\ge 1$. Reject if $\Gamma_L=0$. 
\end{enumerate}
It is worth pointing out that if the population becomes empty at some step, that is, $\Gamma_t=0$, 
then $\Gamma_{t'}=0$ for all $t'>t$ with probability one.  Thus if Arthur observes $\Gamma_t=0$
at some step, he can safely abort the protocol and reject the witness right away. 
  
First let us  check that Arthur can implement the above algorithm in polynomial time
for any Merlin's witness.
Suppose Arthur has already implemented the first $t$ steps of the algorithm
for some $t\ge 0$ and needs to implement the next step. 
Obviously, the number of occupied points  at step $t$ is at most $\Gamma_t$.
Since Arthur has not aborted the algorithm yet, one has  $\Gamma_t\le \Gamma_{max}$ 
and thus there are at most $poly(n)$ occupied points. 
Arthur can store the function $\gamma_t(x)$ efficiently as a list of 
pairs $(x,\gamma_t(x))$ which includes only occupied points $x$. 
Next Arthur needs to generate the function $\gamma_{t+1}$
according to Eq.~(\ref{next}). 
Note that if some point $x$
is empty at step $t$, that is, $\gamma_t(x)=0$, such point does not contribute to the sum in Eq.~(\ref{next})
since $\pois{0}=0$ with probability one. By the same reason, a point $y$ can be
occupied at step $t+1$ with a non-zero probability only if $\langle x|P|y\rangle>0$ for some point $x$ which
is occupied at step $t$. 
The number of such points $y$ is at most $poly(n)$
since $G$ (and thus $P$) has at most $poly(n)$ non-zero matrix
elements in each row.  Hence for a given function $\gamma_t$
there are at most $poly(n)$ Poisson variables that Arthur has to 
generate in order to determine the function $\gamma_{t+1}$. To generate each of those variables
Arthur has to compute $\langle x|P|y\rangle$. This requires computing $\langle x|G|y\rangle$
and the ratio 
$\phi(y)/\phi(y)$. Both computations can be done in time $poly(n)$ since
$G$ is a sum of local operators while $\phi$ is described
by a polynomial-size circuit.
Finally, generating a Poisson random variable with a specified mean 
can be done in constant time. To avoid complications related to approximating the Poisson distribution we shall
assume that Arthur has an access to a device that takes as input a mean  $p$ and outputs a
random non-negative integer drawn from $\pois{p}$.

\subsection{Miscellaneous} 

Here we state some basic facts needed for the proof of completeness and soundness conditions.
Define a non-negative state 
\begin{equation}
\label{phi}
|\phi\rangle=\sum_{x\in \{0,1\}^n} \phi(x)\, |x\rangle,
\end{equation}
where $\phi(x)$ is the regularized version of $\phi_M$ defined in Eq.~(\ref{regularized}).
Note that $\phi$ may or may not be normalized. 
We shall need the first and the second moments of $\Gamma_t$. Here and below
the probability distribution of $\Gamma_t$ is
obtained by iterating
Eq.~(\ref{next}) {\em without} imposing a constraint $\Gamma_t\le \Gamma_{max}$.
\begin{lemma}
\label{lemma:moments}
For any $t=1,\ldots,L$ one has 
\label{lemma:moments}
\begin{equation}
\label{first}
\expect{\Gamma_t}  =\frac1{\phi(x_M)}  \langle x_M|G^t| \phi \rangle,
\end{equation}
and
\begin{equation}
\label{second}
 \expect{\Gamma_L^2}=
 \frac1{\phi(x_M)}\sum_{s=0}^L \sum_{y\in \{0,1\}^n} \frac1{\phi(y)}\cdot 
 \langle x_M|G^s |y\rangle \cdot \langle y|G^{L-s}|\phi\rangle^2.
\end{equation}
\end{lemma}
The proof based on Eq.~(\ref{facts}) is a straightforward calculation, so we 
postpone it until the end of this section.
\begin{lemma}
\label{lemma:guide}
Suppose $H$ admits a guiding state. Then there exists at least one guiding state
$\phi$ such that $\langle x|\phi\rangle \ge 2^{-n-1}$ for all $x\in \{0,1\}^n$.
\end{lemma}
\begin{proof}
Indeed, let $|\omega\rangle=\sum_x \omega(x)\, |x\rangle$ be some guiding state.
By definition,  it means that $\omega(x)\ge 0$ for all $x$, $\sum_x \omega^2(x)=1$,
the function $x\to \omega(x)$ can be computed by a polynomial-size circuit, 
and $H$ has a non-negative normalized ground state $\psi$ such that 
$\langle x|\psi\rangle \le poly(n)\cdot \omega(x)$ for all $x$. Define a function
\[
\phi(x)=C_n(\omega(x)+2^{-n})
\]
where $C_n>0$ is a constant
chosen such that  $\sum_x \phi^2(x)=1$. Simple algebra shows that 
$1-2^{-\Omega(n)} \le C_n \le 1$, that is, $C_n \approx 1$ for  large $n$. 
Thus $\phi(x) \ge 2^{-n-1}$ and $\langle x|\psi\rangle \le C^{-1}_n poly(n) \phi(x)=poly(n) \phi(x)$ for all $x$.
The function $x\to \phi(x)$ has a polynomial-size circuit
since one extra addition and a multiplication by a constant can only
increase the circuit size by $poly(\log{(n)})$. Therefore $|\phi\rangle=\sum_x \phi(x)\, |x\rangle$ is the desired guiding state. 
\end{proof}
%%%%%%%%%%%%%%%%%%%%%%%%%%%%%%%%%%%%%%%%%%%%
\subsection{Proof of soundness}

Consider a no-instance. 
We have to prove that the acceptance probability $P_{acc}$ is small for any witness
$\lambda_M,\phi_M,x_M$. One can get an upper bound on $P_{acc}$ by omitting the tests
$\Gamma_t\le \Gamma_{max}$ since Arthur rejects whenever one of these tests fails.
Thus
\begin{equation}
\label{no1}
P_{acc} \le \prob{ \Gamma_L\ge 1} \le \expect{\Gamma_L}  = \frac1{\phi(x_M)} 
\langle x_M|G^L|\phi\rangle \le \frac{\|\phi\|}{\phi(x_M)} \cdot \|G\|^L.
\end{equation} 
Here we  used Eq.~(\ref{first}). By definition of $\phi$ 
one has $2^{-n-1}\le \phi(x)\le 1$ for all $x\in \{0,1\}^n$, see Eq.~(\ref{regularized}). 
Thus $\|\phi\|\le 2^{n/2}$ and $\phi(x_M)\ge 2^{-n-1}$. 
Furthermore, $\|G^L\| =\|G\|^L \le (1-\Delta)^L$, where 
$\Delta=\beta(\lambda_{no}-\lambda_{yes})$ is the decision gap,
see Eq.~(\ref{Gnorm-no},\ref{Delta}).
Hence 
\begin{equation}
\label{no2}
P_{acc}\le 2^{O(n)} (1-\Delta)^L \le 2^{O(n)} e^{-\Delta L}.
\end{equation}
Here we used a bound $1-s\le e^{-s}$ which holds  for all $s\ge 0$. 
Recall that $\Delta\ge poly(1/n)$.
Hence choosing  $L=\Omega(n\Delta^{-1})=poly(n)$ we can make $P_{acc}\le 2^{-n}$
for any no-instance and for any Merlin's witness. 

%%%%%%%%%%%%%%%%%%%%%%%%%%%%%%%%%%%%%%%%%%%%
\subsection{Proof of completeness}

Consider a yes-instance and Merlin's witness $(\lambda_M,\phi_M,x_M)$,
where $\lambda_M=\lambda$ is the ground state energy and
$\phi_M$ computes amplitudes  of some guiding state satisfying conditions of Lemma~\ref{lemma:guide}.
Then  the regularized version of  $\phi_M$ defined in Eq.~(\ref{regularized})
coincides with $\phi_M$. Hence $|\phi\rangle=\sum_x \phi(x)\, |x\rangle$ is a normalized 
guiding state for $H$. 
We will prove that $P_{acc}\ge poly(1/n)$ for some choice of the string $x_M$
(although the proof is not constructive). Recall that $x_M$ determines
the initial state of the walk, see Eq.~(\ref{init}).
 Indeed,  by definition of the protocol, Arthur accepts iff
$\Gamma_t\le \Gamma_{max}$ for all $t=1,\ldots,L$
and $\Gamma_L\ge 1$. By the union bound,
\begin{equation}
\label{yes1}
1-P_{acc}\le \prob{\Gamma_L=0} + \sum_{t=1}^L \prob{\Gamma_t>\Gamma_{max}}.
\end{equation}
Since $\Gamma_L$ takes non-negative integer values, one can use the second moment bound:
\begin{equation}
\label{yes2}
\prob{\Gamma_L\ge 1} \ge \frac{\expect{\Gamma_L}^2}{\expect{\Gamma_L^2}}.
\end{equation}
By Markov's inequality, $\prob{\Gamma_t>\Gamma_{max}} \le \expect{\Gamma_t}/\Gamma_{max}$. 
Hence
\begin{equation}
\label{yes2}
P_{acc}\ge \frac{\expect{\Gamma_L}^2}{\expect{\Gamma_L^2}} -\frac1{\Gamma_{max}}\sum_{t=1}^L \expect{\Gamma_t}.
\end{equation}
It suffices to show that there exists a string $x_M$  and $poly(n)$ functions $p(n),q(n)$
independent of $\Gamma_{max}$ such that 
\begin{equation}
\label{yes3}
\frac{\expect{\Gamma_L^2}}{\expect{\Gamma_L}^2} \le p(n)
\quad \mbox{and} \quad
\expect{\Gamma_t} \le q(n) \quad \mbox{for all $t=1,\ldots,L$}.
\end{equation}
(Note that all expectation values above depend on $x_M$.)
Indeed, in this case Eq.~(\ref{yes2}) implies
\begin{equation}
\label{yes4}
P_{acc}\ge \frac1{p(n)} - \frac{Lq(n)}{\Gamma_{max}} \ge \frac1{2p(n)}
\end{equation}
if we choose $\Gamma_{max}=2p(n)q(n)L=poly(n)$.
Let us now prove existence of a string $x_M$ satisfying Eq.~(\ref{yes3}).
Since $\phi$ is a guiding state for $H$, there exists  a non-negative normalized ground state 
$|\psi\rangle = \sum_x \psi(x)\, |x\rangle$ such that
\begin{equation}
\label{yes5}
\psi(x)\le r(n)\cdot \phi(x) \quad \mbox{for all $x$},
\end{equation}
where $r(n)\le poly(n)$, see Definition~\ref{dfn:guided}.
Define a set 
\begin{equation}
\label{Sgood}
\calS=\{ x\in \{0,1\}^n \, : \, \frac{\psi(x)}{\phi(x)} \ge \frac{\langle \psi|\phi\rangle}2 \}
\end{equation}
and a probability distribution 
\begin{equation}
\label{pi}
\pi(x)=\frac{\psi(x)\phi(x)}{\langle \psi|\phi\rangle}.
\end{equation}
(One may think of $\pi$ as a ``steady state" of $P$ since $\pi P=\pi$.)
We claim that 
\begin{equation}
\label{Slikely}
\pi(\calS)\equiv \sum_{x\in \calS} \pi(x) \ge \frac12.
\end{equation}
Indeed, one has
\begin{equation}
\label{Slikely1}
1 = \sum_x \pi(x)
=\pi(\calS) + \langle\psi|\phi\rangle^{-1} \sum_{x\notin \calS} \phi^2(x) \cdot \frac{\psi(x)}{\phi(x)}
 \le \pi(\calS) + \frac12 \sum_{x\notin \calS} \phi^2(x)  \le \pi(\calS) +\frac12.
\end{equation}
Here the last bound uses normalization of $\phi$. This proves Eq.~(\ref{Slikely}). 
For each $t=1,\ldots,L$ define a function 
\[
V_t(x)=\frac{\langle x|G^t|\phi\rangle}{\phi(x)}.
\]
From Eq.~(\ref{first}) one infers that $\expect{\Gamma_t}=V_t(x_M)$. 
Furthermore,
\begin{equation}
\label{yes6}
\sum_x \pi(x) V_t(x) = \langle\psi|\phi\rangle^{-1} \sum_x \psi(x) \langle x|G^t|\phi\rangle=1
\end{equation}
since $G^t\psi=\psi$. Next define a function 
\[
W(x)= \sum_{s=0}^L \sum_{y} \frac{\langle x|G^s|y\rangle \cdot \langle y| G^{L-s} |\phi\rangle^2}{\phi(x)\phi(y)}.
\]
Here the second sum is over all $y\in \{0,1\}^n$. 
From  Eq.~(\ref{second}) one infers that $\expect{\Gamma_L^2}=W(x_M)$.
Taking into account that $G^s\psi=\psi$ for any $s$ one gets
\begin{equation}
\label{yes7}
\sum_x \pi(x) W(x) =
\langle \psi|\phi\rangle^{-1} \sum_{s=0}^L  \sum_{y}  \frac{\psi(y)}{\phi(y)}\cdot   \langle y| G^{L-s} |\phi\rangle^2.
\end{equation}
Note that $\psi(y)/\phi(y)\le r(n)$ and
$\langle \psi|\phi\rangle=\sum_x \psi^2(x) \phi(x)/\psi(x) \ge 1/r(n)$
due to Eq.~(\ref{yes5}).
Since $G$ has non-negative matrix elements,  $\|G\|=1$, and $\|\phi\|=1$ we arrive at
\begin{equation}
\label{yes8}
\sum_x \pi(x) W(x)  \le r^2(n) \sum_{s=0}^L \sum_{y}   \langle y| G^{L-s} |\phi\rangle^2
=r^2(n) \sum_{s=0}^L \langle \phi | G^{2(L-s)}|\phi\rangle \le r^2(n)(1+L).
\end{equation}
Combining Eqs.~(\ref{Slikely},\ref{yes6},\ref{yes8})  results in 
\begin{equation}
\label{yes9}
\frac1{\pi(\calS)} \sum_{x\in \calS} \pi(x) \left[ W(x) + \sum_{t=1}^L V_t(x) \right] \le 2(L+1)(1+r^2(n))\equiv q(n).
\end{equation}
Therefore there must exist $x_M\in \calS$ such that 
\begin{equation}
\label{yes10}
W(x_M)\le q(n) \quad \mbox{and} \quad V_t(x_M)\le q(n) \quad \mbox{for all $t=1,\ldots,L$}.
\end{equation}
Furthermore, for any $x\in \calS$ one has
\begin{equation}
\label{yes11}
V_L(x)=\frac{\langle x|G^L|\phi\rangle}{\phi(x)}\ge \frac{\langle x|G^L|\psi\rangle}{r(n)\phi(x)}
=\frac{\psi(x)}{r(n) \phi(x)} \ge \frac{\langle \psi|\phi\rangle}{2r(n)} \ge \frac1{2r^2(n)}.
\end{equation}
This shows that there exists $x_M\in \calS$ such that 
\begin{equation}
\label{yes12}
\frac{\expect{\Gamma_L^2}}{\expect{\Gamma_L}^2}
=\frac{W(x_M)}{V_L(x_M)^2} \le 4r^4(n)q(n)\equiv p(n)
\quad \mbox{and} \quad \expect{\Gamma_t} \le q(n) \quad \mbox{for all $t=1,\ldots,L$}.
\end{equation}
This proves the desired bounds in Eq.~(\ref{yes3})
and shows that Merlin can make Arthur to accept with probability
at least $P_{acc}\ge 1/2p(n)\ge poly(1/n)$ provided that 
$\Gamma_{max}=\Omega(L^3 r^8(n))$.

Arthur can achieve the acceptance probability at least $2/3$ as required 
for the completeness condition  by implementing 
about $2p(n)$ independent rounds of the above algorithm
with the same witness. Arthur accepts the witness iff at least one of the rounds 
outputs `accept'. 
 For any no-instance
 each round accepts with probability at most $2^{-n}$ 
 and thus the full algorithm accepts with probability less than $1/3$ for large enough $n$.
This completes the proof that Guided Stoq-LH is contained in promise-$\MA$.

Finally, the statement that Guided Stoq-LH is complete for promise-$\MA$ if $k\ge 6$
follows trivially from Ref.~\cite{BDOT06}. Indeed, let $\calL=\calL_{yes}\cup \calL_{no}$
be any language in promise-$\MA$, where $\calL_{yes}$ and $\calL_{no}$ are the sets of
yes- and no-instances. Without loss of generality we can assume that Arthur's verification
protocol has perfect completeness~\cite{Furer1989}, that is, 
for any yes-instance $\calI\in \calL_{yes}$ there exists a witness $\calW$
such that $P_{acc}(\calI,\calW)=1$. Using Lemma~3 of  Ref.~\cite{BDOT06}
one can efficiently transform any instance $\calI\in \calL$ into a  $6$-local stoquastic Hamiltonian $H$ 
such that its ground state energy satisfies $\lambda=0$ if $\calI\in \calL_{yes}$
and $\lambda\ge poly(1/n)$ if $\calI\in \calL_{no}$. Furthermore, for any yes-instance $\calI$
one can choose a ground state of $H$ as 
a coherent superposition of all computational branches of Arthur's verification algorithm
that lead to the acceptance (for some fixed witness $\calW$ such that $P_{acc}(\calI,\calW)=1$).
Hence computing the amplitudes of $\psi$ is equivalent 
to checking  whether a sequence of binary strings represent a valid computational path 
of the algorithm. This can be checked in time $poly(n)$. Therefore for a yes-instance 
the ground state $\psi$ itself can be chosen as a guiding state. 
This shows that any problem in promise-$\MA$ can be reduced to Guided Stoq-LH with $k=6$.

\subsection{Open problems}
\label{subs:problems1}

A natural question is whether the above algorithm can be used to estimate the ground state energy without Merlin's 
assistance assuming that one has a good guess of the guiding state. 
One possible strategy would be to sweep $\lambda_M$ over a region that is likely to contain
the ground state energy and estimate Arthur's  acceptance probability $P_{acc}$
for each value of $\lambda_M$  by Monte Carlo simulation.  One should expect that $P_{acc}$ is negligible
unless $\lambda_M\approx \lambda$. Indeed, recall that $\|G\|=1-\beta(\lambda-\lambda_M)$. 
If $\lambda_M>\lambda$ then $\|G\|>1$ which leads to an exponential growth 
of the population, see Eq.~(\ref{first}). Accordingly, 
 the test $\Gamma_t\le \Gamma_{max}$ is likely to fail.
On the other hand, if $\lambda_M<\lambda$ one has $\|G\|<1$ and 
the final population is likely to be empty.  In both cases the outcome of the protocol is `reject'.
Unfortunately, proving that $P_{acc}$ is non-negligible for $\lambda\approx \lambda_M$ 
requires a careful choice of the initial string $x_M$.
A preliminary analysis shows that $x_M$ can be chosen efficiently
if the guiding state obeys a stronger condition 
$poly(1/n) \langle x|\psi\rangle \le \langle x|\phi\rangle \le poly(n)\langle x|\psi\rangle$
for all $x$.

One may also ask  whether Theorem~\ref{thm:guided} holds for some weaker
notion of a guiding state. For example, the pointwise correlation condition
in Eq.~(\ref{guiding}) appears to be unreasonably  strong. It would be very desirable 
to replace Eq.~(\ref{guiding})   by a bound on some global  correlation measure that has a clear physical meaning. 
Ideally, a guiding state $\phi$ just needs to have
an overlap $\ge poly(1/n)$ with some exact ground state $\psi$ and have efficiently computable
amplitudes. A  preliminary analysis shows that for Hamiltonians with a polynomial spectral gap
a guiding state only needs to satisfy a condition $\sum_x \langle x|\psi\rangle^3/\langle x|\phi\rangle \le poly(n)$.
Alternatively,  one can study guiding states
$\phi$ that admit an efficient classical algorithm for sampling a basis vector $x$
from the distribution $\langle x|\phi\rangle^2$.

Finally, it is important to identify non-trivial classes of stoquastic Hamiltonians that 
actually admit a guiding state. For example, results of Ref.~\cite{BT07}
imply that for any frustration-free stoquastic Hamiltonian
one can always choose a non-negative ground state with efficiently computable amplitudes. 
In this case the ground state itself can serve as a guiding state.
We anticipate that non-trivial examples of guiding states could be found 
among tensor network states such as Matrix Product States or PEPS~\cite{Verstraete2008}.
We conjecture that the ferromagnetic TIM Hamiltonians studied in the next section admit a guiding state.

\subsection{Proof of Lemma~\ref{lemma:moments}}
For any integers $0\le t\le s$ 
define a random variable
\begin{equation}
\label{moments1}
\Gamma_{t,s}=\left\{ \begin{array}{rcl}
\Gamma_t &\mbox{if} & s=0, \\
\sum_{y,z} \gamma_t(y) \langle y|P^s|z\rangle &\mbox{if}& s\ge 1. \\
\end{array}\right.
\end{equation}
Applying repeatedly Eq.~(\ref{facts}) one can easily show that
\begin{equation}
\label{moments2}
\expect{\Gamma_{t,s}}=\expect{\Gamma_{t-1,s+1}}.
\end{equation}
Note that $\Gamma_{0,t}$ is a deterministic variable,
\begin{equation}
\label{moments3}
\Gamma_{0,t}=\sum_z \langle x_M|P^t|z\rangle.
\end{equation}
Hence
\begin{equation}
\label{moments4}
\expect{\Gamma_t}=\expect{\Gamma_{t,0}}=\expect{\Gamma_{0,t}}=\sum_z \langle x_M|P^t|z\rangle=
\frac1{\phi(x_M)} \langle x_M|G^t|\phi\rangle.
\end{equation}
This proves Eq.~(\ref{first}).
  Let us now compute $\expect{\Gamma_{t,s}^2}$ with respect to the conditional
distribution 
\[
\prob{\gamma_t|\gamma_1,\ldots,\gamma_{t-1}}=\prob{\gamma_t|\gamma_{t-1}}.
\]
For a fixed $\gamma_{t-1}$ the variables $\gamma_t(y)$ are independent Poisson-distributed variables
so that 
\begin{equation}
\label{tricky1}
\expect{\gamma_t(y)}=\sum_x  \gamma_{t-1}(x) \langle x|P|y\rangle
\end{equation}
and
\begin{equation}
\label{tricky2}
\expect{\gamma_t(y)\gamma_t(y')}=\expect{\gamma_t(y)}\cdot \expect{\gamma_t(y')} + \delta_{y,y'} \expect{\gamma_t(y)}.
\end{equation}
Using Eq.~(\ref{tricky1},\ref{tricky2}) one easily gets
\begin{equation}
\label{tricky3}
\expect{\Gamma_{t,s}^2}=\Gamma_{t-1,s+1}^2 + 
\sum_{x,y,z,z'} \gamma_{t-1}(x)  \langle x|P|y\rangle \cdot \langle y | P^s | z\rangle \cdot \langle y | P^s| z'\rangle.
\end{equation}
Taking the expectation value over the distribution of $\gamma_1,\ldots,\gamma_{t-1}$ leads to
\begin{equation}
\label{tricky4}
\expect{\Gamma_{t,s}^2}=\expect{\Gamma_{t-1,s+1}^2} 
+ \sum_{y,z,z'} \langle x_M|P^t |y\rangle \cdot  \langle y|P^s|z\rangle \cdot   \langle y|P^s|z'\rangle.
\end{equation}
Iterating Eq.~(\ref{tricky4}) starting from $t=L$, $s=0$ and using Eq.~(\ref{moments3})
one gets
\begin{equation}
\label{tricky6}
\expect{\Gamma_L^2}=\sum_{\substack{s+t=L\\ s,t\ge 0 \\}}\; \sum_y \langle x_M|P^t|y\rangle 
\left[ \sum_z \langle y|P^s|z\rangle \right]^2.
\end{equation}
Substituting the definition of $P$, see Eq.~(\ref{P(x,y)}), leads to Eq.~(\ref{second}).
This proves Lemma~\ref{lemma:moments}.

\section{Approximating the partition function of TIM}

In this section we prove Theorem~\ref{thm:TIM}.
We shall use a notation $\rho(A)$ for the difference
between the largest and the smallest eigenvalue of a hermitian matrix $A$.
Note that $\rho(A)\le 2\|A\|$ for any matrix $A$.
The following lemma will be  needed to control the error in the Suzuki-Trotter 
approximation. 
\begin{lemma}
\label{lemma:ABD}
Consider any pair of hermitian operators $A,B$ and let $\rho=\rho(A)+\rho(B)$.
For any $0\le t\le (2\rho)^{-1}$ there exists a hermitian
operator $D$ such that  $\|D\|\le 12 \rho^3$ and
\begin{equation}
\label{lemma1}
e^{At/2} e^{Bt} e^{At/2} = e^{(A+B)t+Dt^3}.
\end{equation}
\end{lemma}
Since the proof involves a straightforward  calculation, we postpone it 
until the end of this section.  Consider now a TIM Hamiltonian
\begin{equation}
\label{TIM2}
H=-A-B, \quad A=\sum_{1\le u<v\le n} J_{u,v} Z_u Z_v, \quad B=\sum_{1\le u\le n} h_u X_u,
\end{equation}
where $J_{u,v}\ge 0$ for all $u,v$ (the ferromagnetic case). Let $\calZ=\trace{\left(e^{A+B}\right)}$
be the partition function. 
For any integer $r\ge 1$ define a Suzuki-Trotter approximation to $\calZ$ as 
\begin{equation}
\label{ST}
\calZ'=\trace{\, \left( e^{At} e^{Bt} \right)^r }, \quad t\equiv r^{-1}.
\end{equation}
Note that $\rho\equiv \rho(A)+\rho(B)\le 2(\|A\|+\|B\|)\le poly(n,J)$. 
Suppose  $r\ge 2\rho$ so that $0\le t\le (2\rho)^{-1}$. Then Lemma~\ref{lemma:ABD} implies
\begin{equation}
\label{Z'}
\calZ'= \trace{\, \left(e^{At/2} e^{Bt} e^{At/2} \right)^r }=\trace{\, e^{r\left[ (A+B)t + Dt^3\right]} } =
\trace{\, e^{A+B+C} },
\end{equation}
where $C=Dr^{-2}$ is a hermitian operator such that $\|C\|\le 12\rho^3 r^{-2}$.

 Let $\lambda_i$ and $\lambda_i'$ be the $i$-th largest eigenvalue
of the Hamiltonian $A+B$ and $A+B+C$  respectively, where $i=1,\ldots,2^n$. By Weyl's inequality and Lemma~\ref{lemma:ABD},
\begin{equation}
\label{Weyl}
|\lambda_i-\lambda_i'|\le \|C\|\le \frac{12\rho^3}{r^2}.
\end{equation}
Choosing
\begin{equation}
\label{r}
r=\max{  \left[ 2\rho, \sqrt{12\rho^3 \delta^{-1} } \right] }
\end{equation}
guarantees that conditions of Lemma~\ref{lemma:ABD} are satisfied and 
$\|C\|\le \delta$. Note that $r\le poly(n,J,\delta^{-1})$ since
$\rho\le poly(n,J)$. Then $|\lambda_i'-\lambda_i|\le \delta$ and
\begin{equation}
\calZ'=\sum_{i=1}^{2^n} e^{\lambda_i'} \le  \sum_{i=1}^{2^n} e^{\lambda_i+\delta} \approx (1+\delta) \calZ.
\end{equation}
Here we assumed for simplicity that $\delta\ll 1$.
The same arguments show that
$\calZ'\ge (1-\delta) \calZ$. Hence $\calZ'$ approximates $\calZ$ up to a multiplicative error $\delta$.

The remaining step is  the standard quantum-to-classical mapping that 
relates $\calZ'$ to the partition function of a classical Ising model~\cite{Sachdev2007}.
Let $\sigma=(\sigma_1,\ldots, \sigma_n)\in \{\pm 1\}^n$ be a configuration of  $n$ classical Ising spins
and $|\sigma\ra$ be the corresponding basis state of $n$ qubits. Then
\begin{equation}
\label{eAt}
e^{At}=\sum_{\sigma \in \{\pm 1\}^n} \; e^{E_A(\sigma)} |\sigma\rangle\langle\sigma|
\quad \mbox{where} \quad
E_A(\sigma)=\sum_{1\le u<v\le n} tJ_{u,v} \, \sigma_u \sigma_v
\end{equation}
and
\begin{equation}
\label{eBt}
e^{Bt} = \Gamma \sum_{\sigma,\sigma'\in \{\pm 1\}^n}\;  e^{E_B(\sigma,\sigma')} |\sigma\rangle\langle\sigma'|
\quad \mbox{where} \quad
E_B(\sigma,\sigma')=\sum_{1\le u\le n} \tilde{h}_u \sigma_u \sigma_u'.
\end{equation}
Here
\begin{equation}
\label{tildeh}
\tilde{h}_u = -\frac12 \log{\left[ \tanh{(t h_u)} \right] } \quad \mbox{and} \quad
\Gamma=2^{-n/2} \prod_{u=1}^n \sqrt{\sinh{(2t h_u)}}.
\end{equation}
Note that $\tilde{h}_u\ge 0$ since we assumed $h_u\ge 0$.
From Eqs.~(\ref{eAt},\ref{eBt}) one gets
\begin{equation}
\label{QtoC}
\calZ'=\trace{\, \left( e^{At} e^{Bt} \right)^r }= \Gamma^r 
\sum_{\sigma^1,\ldots,\sigma^r\in \{\pm 1\}^n} \; \exp{\left[ \sum_{i=1}^r E_A(\sigma^i) + E_B(\sigma^i,\sigma^{i+1})\right]}
\equiv \sum_{\theta\in \{\pm 1\}^{nr}} e^{E(\theta)}.
\end{equation}
The energy function $E(\theta)$ describes $r$ copies of the $n$-spin ferromagnetic Ising model
such that the $i$-th  copy has  the energy function $E_A(\sigma^i)$ and  each consecutive pair of copies is coupled 
by the energy function $E_B(\sigma^i,\sigma^{i+1})$.
The multiplicative factor $\Gamma^r$ can be absorbed into a constant energy shift in $E(\theta)$. 
Since this factor is easy to compute, for simplicity we shall ignore it. 
Hence $\calZ'$ is the partition function of a classical ferromagnetic Ising model.
We can now invoke Theorem~5 of Ref.~\cite{JS93}. It asserts that 
$\calZ'$ admits FPRAS with a running time
 $O(\delta^{-2}M^3 N^{11} \log{N})$, where $N=nr$ is the total number of spins
and $M$ is the number of non-zero spin-spin couplings. Note that in our case
$M\le nr+{n \choose 2}r \le n^2r$.
Using Eq.~(\ref{r}) with a conservative estimate $\rho\le n^2 J$ one gets
$r\le  n^3 J^{3/2} \delta^{-1/2}$. Thus $N\le n^4 J^{3/2} \delta^{-1/2}$,  $M\le n^5 J^{3/2} \delta^{-1/2}$
and the FPRAS for $\calZ'$ has running time $O(n^{59} J^{21} \delta^{-9})$, ignoring logarithmic factors. 
Since $\calZ'$ approximates $\calZ$ with a multiplicative error $\delta$,
 this proves Theorem~\ref{thm:TIM}.

{\em Remark:} We note that $\tilde{h}_u$ becomes infinite if $h_u=0$.
One can always assume that $h_u\ge \delta/n$ since changing the Hamiltonian
by a perturbation of norm at most $\delta$ leads to a multiplicative error of order $\delta$
in the partition function.

\subsection{Proof of Lemma~\ref{lemma:ABD}}
Since Eq.~(\ref{lemma1}) is invariant under a shift $A\to A+cI$, we can
assume that the largest eigenvalue of $A$ is zero. Then $A\le 0$
and $\|A\|=\rho(A)$. 
Likewise, we can assume that $B\le 0$ and $\|B\|=\rho(B)$. 
Consider the Taylor series
\[
Z(t)\equiv e^{At/2} e^{Bt} e^{At/2}=\sum_{p=0}^\infty Z_p\, t^p
\]
that converges absolutely for any $t\in \CC$.
One can easily check that
\[
Z_0=I, \quad Z_1=A+B, \quad Z_2=Z_1^2/2.
\]
We shall upper bound the norm of higher order coefficients using the Cauchy's formula,
\[
Z_p=\frac1{2\pi i} \oint_{|t|=R} \frac{Z(t)dt}{t^{p+1}}.
\] 
Since $Z(t)$ is analytic in the full complex plane, the radius $R$ can be chosen 
arbitrarily. Let us choose $R=\rho^{-1}$ and 
let $C$ be the maximum of $\|Z(t)\|$ over the circle  $|t|=R$.
Then $\|Z_p\| \le CR^{-p}=C\rho^p$ and 
\[
Z(t)=\sum_{p=0}^2 Z_p t^p + \Delta \quad \mbox{where} \quad \|\Delta\|\le C\sum_{p\ge 3}
(\rho t)^p  \le 2C\rho^3 t^3.
\]
Here we used the assumption that $0\le t\le (2\rho)^{-1}$.  
Note that 
\[
\|Z(t)\| \le \|e^{At/2}\|^2 \cdot \|e^{Bt}\| \le e^{(\|A\| + \|B\|)|t|}=e^{\rho |t|} = e.
\]
for any complex $t$ with $|t|=R=\rho^{-1}$. Hence $C\le e$.

Choose any $0\le t\le (2\rho)^{-1}$ and define
\[
U(t)\equiv Z(t)-I=Z_1 t + Z_2 t^2 + \Delta.
\]
Taking into account that $e^{At/2}\le I$ and $e^{Bt}\le I$ one easily gets 
\begin{equation}
\label{U(t)bound}
\|U(t)\| \le 2\|e^{At/2} - I \| + \| e^{Bt}-I\| \le  t(\|A\| + \|B\|) =\rho t \le \frac12.
\end{equation}
Here we used the inequality $e^x\ge 1+x$ which holds for $x\le 0$.
 One can easily  check that 
\begin{equation}
\label{C'}
\| \log{(I+U)} - U + U^2/2 \| \le C' \|U\|^3,
\end{equation}
for any hermitian operator $U$ such that $\|U\|\le 1/2$, 
where 
\[
C'\equiv \max_{x\, : \, |x|\le 1/2} \; \; \frac1{|x|^3} \left| \log{(1+x)} - x + x^2/2 \right| \approx 0.55.
\]
Choose the desired operator $D$ as
\begin{equation}
\label{C}
D=t^{-3}\left[ -(A+B)t + \log{\left( e^{At/2} e^{Bt} e^{At/2}\right)} \right]
=t^{-3}\left[ -Z_1 t + \log{(I+U(t))}\right].
\end{equation}
By definition, $D$ is hermitian. Combining Eq.~(\ref{U(t)bound},\ref{C'}) 
and using the union bound one  gets
\[
\| D\| \le t^{-3} \| -Z_1 t + U(t) - U(t)^2/2 \| + C' \rho^3.
\]
A simple algebra shows that
\[
\Gamma\equiv -Z_1 t + U(t) - U(t)^2/2 =  \Delta -\frac{Z_1^4 t^4}8 - \frac{\Delta^2}2 - \frac{Z_1^3 t^3}2 - Z_1 \Delta t -
 \frac{Z_1^2 \Delta t^2}2.
\]
Taking into account that $\|Z_1\|\le \rho $ and $\rho t \le 1/2$ one arrives at
\[
\|\Gamma\| \le (\rho t)^3\left(  2C + \frac{1}{16} + \frac{C^2}4 + \frac{1}2 + C +\frac{C} 4 \right)\le 
(\rho t)^3 \left( \frac9{16} + \frac{13C}4 + \frac{C^2}4\right).
\]
Therefore
\[
\|D\| \le \rho^3 \left( \frac9{16} + \frac{13C}4 + \frac{C^2}4 + C'\right).
\]
Substituting $C\le e$ and $C'\le 0.56$ leads to $\|D\|\le 12 \rho^3$. 

\vspace{10mm}

{\bf Acknowledgments --}
The author would like to thank Toby Cubitt and Graeme Smith for  helpful discussions.
This work was supported in part  by the DARPA QuEST program under contract number HR0011-09-C-0047.

%\bibliographystyle{unsrt}
%\bibliography{mybib}

%\end{document}

\end{document}